\documentclass[reprint,amsmath,amssymb,aps,pre]{revtex4-1}

\usepackage{graphicx}  
\usepackage{dcolumn}   
\usepackage{bm}        
\usepackage{amssymb}   
\usepackage{mathtools}
\usepackage{amsmath}
\usepackage{color}
\usepackage{dblfloatfix}
\usepackage{comment}
\usepackage{amssymb,amsthm,amsfonts,amscd,mathrsfs}
\usepackage{thm-restate}
\usepackage{amsmath}
\usepackage{natbib}
\usepackage{paralist}

\usepackage{xcolor}         % colors

\newtheorem{thm}{Theorem}
\newtheorem{prop}[thm]{Proposition}

\begin{document}

\title{Frequency propagation: Multi-mechanism learning in nonlinear physical networks}

\author{Vidyesh Rao Anisetti}
\email{vvaniset@syr.edu}
\affiliation{Physics Department, Syracuse University, Syracuse, NY 13244 USA}
\author{Ananth Kandala}
\email{an.kandala@ufl.edu}
\affiliation{Department of Physics, University of Florida, Gainesville, FL 32611 USA}
\author{Benjamin Scellier}
\email{benjamin.scellier@polytechnique.edu}
\affiliation{Department of Mathematics, ETH Z\"{u}rich, 8092 Z\"{u}rich, Switzerland}
\author{J.~M.~Schwarz}
\email{jmschw02@syr.edu}
\affiliation{Physics Department, Syracuse University, Syracuse, NY 13244 USA}
\affiliation{Indian Creek Farm, Ithaca, NY 14850 USA}

\begin{abstract}
We introduce \textit{frequency propagation}, a learning algorithm for nonlinear physical networks. In a resistive electrical circuit with variable resistors, an activation current is applied at a set of input nodes at one frequency, and an error current is applied at a set of output nodes at another frequency. The voltage response of the circuit to these boundary currents is the superposition of an `activation signal' and an `error signal' whose coefficients can be read in different frequencies of the frequency domain. Each conductance is updated proportionally to the product of the two coefficients. The learning rule is local and proved to perform gradient descent on a loss function.
We argue that frequency propagation is an instance of a \textit{multi-mechanism learning} strategy for physical networks, be it resistive, elastic, or flow networks. Multi-mechanism learning strategies incorporate at least two physical quantities, potentially governed by independent physical mechanisms, to act as activation and error signals in the training process. Locally available information about these two signals is then used to update the trainable parameters to perform gradient descent. We demonstrate how earlier work implementing learning via \textit{chemical signaling} in flow networks~\cite{anisetti2022learning} also falls under the rubric of multi-mechanism learning.
\end{abstract}

\maketitle

\section{Introduction}

Advancements in artificial neural networks (ANN) ~\cite{goodfellow2016deep} have inspired a search for adaptive physical networks that can be optimized to achieve desired functionality~\cite{anisetti2022learning,kendall2020training,stern_arinze_perez_palmer_murugan_2020,stern2021supervised,lopez2021self,dillavou2021demonstration,scellier2022agnostic,stern_review}. Similarly to ANNs, adaptive physical networks modify their learning degrees of freedom to approximate a desired input-to-output function ; but unlike ANNs, they achieve this using physical laws. In a physical network, the input is typically an externally applied boundary condition, and the output is the network's response to this input, or a statistic of this response. For instance, in a resistive network, an input signal can be fed in the form of applied currents or voltages, and the output may be the vector of voltages across a subset of nodes of the network. The learning degrees of freedom of the network are, for example, the conductances of the resistors (assuming variable resistors). Ideally, these learning parameters must be updated using only locally available information, so that the network can learn without the need for an external supervisor beyond what occurs at the input and the output nodes.
%also called autonomous machines, or self-learning machines.
Moreover, these parameter updates should preferably follow the direction of gradient descent in the loss function landscape, as is the case for ANNs.

Existing learning algorithms for adaptive physical networks include \textit{equilibrium propagation} \citep{kendall2020training,scellier2021deep} and variants of it \citep{stern2021supervised}. These algorithms are based on the idea of \textit{contrastive learning} \cite{baldi1991contrastive} and proceed as follows. In a first phase, an input is presented to the network, either in the form of boundary currents or voltages, and the network is allowed to settle to equilibrium (the `free state'), where a supervisor checks the output of the system. Then the supervisor \textit{nudges} the output
%externally
towards the desired output. This perturbation causes the system to settle to a new (`perturbed') equilibrium state, which is a slightly more accurate approximation of the function that one wants to learn. The supervisor then compares the perturbed state with the free state to make changes in the learning degrees of freedom in such a way that the network spontaneously produces an output that is slightly closer to the desired output.
%In the case of \textit{equilibrium propagation}, it has been shown that such a 
In the limit of infinitesimal nudging, this procedure performs gradient descent on the squared prediction error \cite{scellier2017equilibrium}.

The above procedure requires storing the free state, to compare it with the perturbed state ; in \cite{dillavou2021demonstration}, this is achieved by using two copies of the same network. In this work, we propose an alternative \textit{multi-mechanism learning} approach to overcome this hurdle. It incorporates two physical quantities, each driven by their own respective mechanisms: one quantity acting as an \textit{activation} signal and the other acting as an \textit{error} signal. This concept is motivated by biological systems implementing functionality via multiple biophysical routes or mechanisms. Such functionality can be chemical, electrical or even mechanical in nature with potential interactions between such mechanisms. For instance, in the brain, activity can propagates from one cell to another via electrical and chemical synapses, as opposed to just one mechanism, if you will~\cite{pereda2014}. Given this modularity in functionality in biology, it would be remiss not to explore such richness in how adaptive physical networks learn.
Alternatively, as we shall soon see, the modularity is not necessarily in terms of mechanical versus chemical versus electrical signals, but distinguishable
%/orthogonal
signals.

We introduce \textit{frequency propagation} (Freq-prop), a physical learning algorithm falling under the umbrella concept of multi-mechanism learning. In Freq-prop, the activation and error signals are both sent through a single channel, but are encoded in different frequencies of the frequency domain ; we can thus obtain the respective responses of the system through frequency (Fourier) decomposition. This algorithm, which we show to perform gradient descent, can be used to train adaptive non-linear networks such as resistor networks, elastic networks and flow networks. Freq-prop thus has the potential to be an all-encompassing approach.

\section{Nonlinear Resistive Networks}
\label{sec:resistive-network}

A resistive network is an electrical circuit of nodes interconnected by resistive devices, which includes linear resistors and diodes. Let $N$ be the number of nodes in the network, and denote $v_j$ the electric potential of node $j$. A subset of the nodes are \textit{input nodes}, where we can set input currents: we denote $x_j$ the input current at input node $j$. For each pair of nodes $j$ and $k$, we denote $\theta_{jk}$ the conductance of the resistor between these nodes (provided that the corresponding branch contains a linear resistor). We further denote $\theta = \{ \theta_{jk} : \text{linear branch } (j,k) \}$ the vector of conductances, and $x = (x_1, x_2, \ldots, x_N)$ the vector of input currents, where by convention $x_j=0$ if node $j$ is not an input node. Finally, we denote $v = (v_1, v_2, \ldots, v_N)$ the configuration of the nodes' electric potentials, and $v(\theta,x)$ the equilibrium value of $v$ as a function of the branch conductances ($\theta$) and the input currents ($x$).

The following result, known since the work of Millar \cite{millar1951cxvi}, provides a characterization of the equilibrium state -- see also \cite{kendall2020training} for a proof of this result with notations closer to ours.

\begin{thm}
\label{thm:co-content}
There exists a function $E(\theta,x,v)$ such that
\begin{equation}
    \label{eq:free-state}
    v(\theta,x) = \underset{v}{\arg \min} \; E(\theta,x,v).
\end{equation}
Furthermore, $E(\theta,x,v)$ is of the form
%\begin{equation}
%    \label{eq:energy-function}
%    E(\theta,x,v) = x \cdot v + \frac{1}{2} \sum_{j,k} \theta_{jk} \left( v_j-v_k \right)^2 + E_{\rm nonlinear}(v),
%\end{equation}
\begin{align}
    \label{eq:energy-function-1}
    E(\theta,x,v) & = E_{\rm input}(x,v) + E_{\rm nonlinear}(v) \\
    \label{eq:energy-function-2}
    & + \sum_{{\rm linear \; branch \;} (j,k)} \frac{1}{2} \theta_{jk} \left( v_j-v_k \right)^2,
\end{align}
where $E_{\rm input}(x,v)$ is a function of $x$ and $v$, and $E_{\rm nonlinear}(v)$ is a function of $v$ only.
\end{thm}

$E(\theta,x,v)$ is the `energy function' of the system, also called the \textit{co-content} \cite{millar1951cxvi}, and the equilibrium state $v(\theta,x)$ is a minimizer of the energy function. The energy function contains an energy term $E_{\rm input}(x,v)$ associated to boundary input currents $x$. It also contains energy terms of the form $\theta_{jk} \left( v_j-v_k \right)^2$ representing the power dissipated in branch $(j,k)$. The term $E_{\rm nonlinear}(v)$ contains all nonlinearities of the system. 
In a \textit{linear} resistance network (i.e. when $E_{\rm nonlinear}(v) = 0$), it is well known that the equilibrium configuration of node electric potentials minimizes the power dissipation ;
%this is the \textit{principle of minimum dissipated power}
Theorem \ref{thm:co-content} generalizes this result to nonlinear networks.
Below we explain how the different terms of $E(\theta,x,v)$ are constructed.

\smallskip

\textbf{Constructing the energy function.}
Each branch is characterised by its current-voltage characteristic $i_{jk}=\widehat{i}_{jk}(v_j-v_k)$, where $\widehat{i}_{jk}(\cdot)$ is a real-valued function that returns $i_{jk}$, the current flowing from $j$ to $k$ in response to the voltage $v_j-v_k$. The energy term corresponding to branch $(j,k)$, called the \textit{co-content} of the branch \cite{millar1951cxvi}, is by definition
\begin{equation}
    E_{jk}(v_j-v_k) = \int_0^{v_j-v_k} \widehat{i}_{jk}(v^\prime)dv^\prime.
\end{equation}
In general, the characteristic function $\widehat{i}_{jk}(\cdot)$ is arbitrary, i.e. \textit{nonlinear}. However, some branches are \textit{linear}, meaning that their current-voltage characteristic is of the form $i_{jk} = \theta_{jk} \left( v_j-v_k \right)$, where $\theta_{jk}$ is the branch conductance \footnote{To avoid any confusion, we stress that $\theta_{jk}$ is a scalar, whereas $\hat{i}_{jk}(\cdot)$ is a real-valued function. Thus, $\theta_{jk} \left( v_j-v_k \right)$ denotes the product of $\theta_{jk}$ and $v_j-v_k$, whereas $\hat{i}_{jk}(v_j-v_k)$ denotes the function $\hat{i}_{jk}$ applied to the voltage $v_j-v_k$.}.
For such linear branches, the energy term is
\begin{equation}
    E_{jk}(v_j-v_k) = \frac{1}{2} \theta_{jk} \left( v_j-v_k \right)^2,
\end{equation}
which is the power dissipated in branch $(j,k)$.

We gather all the energy terms of nonlinear branches under a unique term:
\begin{equation}
    E_{\rm nonlinear}(v) = \sum_{\text{nonlinear branch } (j,k)} E_{jk}(v_j-v_k),
\end{equation}
where we recall that $v = (v_1, v_2, \ldots, v_N)$.

As for the energy term $E_{\rm input}(x,v)$, we present two ways to impose boundary conditions to the network to feed it with input signals $x$, either in the form of boundary currents or boundary electric potentials.
%We write $x = (x_j : j \in \{ \text{input nodes} \})$ the vector of input signals.
Recall that we write $x = (x_1, x_2, \ldots, x_N)$ the vector of input signals, where $x_j=0$ if node $j$ is not an input node.
In the case of boundary currents, the corresponding energy term is
\begin{equation}
    E_{\rm input}^{\rm current}(x,v) = \sum_{j \in \{ \text{input nodes} \}} x_j v_j,
\end{equation}
whereas in the case of boundary electric potentials, the energy term is
\begin{equation}
E_{\rm input}^{\rm voltage}(x,v) = 
\left\{
\begin{array}{l}
    \displaystyle 0 \quad \text{if} \quad v_j = x_j, \; \forall j \in \{ \text{input nodes} \}, \\
   \displaystyle +\infty \quad \text{otherwise},
\end{array}
\right.
\end{equation}
i.e. the electric potential $v_j$ is clamped to $x_j$ for every input node $j$ (so that the energy remains finite).
%We may also see $E_{\rm input}^{\rm voltage}(x,v)$ as $\sum_{j \in \{ \text{input nodes} \}} G (x_j - v_j)^2$ for some large constant $G \gg 1$: in the limit of infinite conductance, $G \to \infty$, the pressure $v_j$ has to be clamped to $x_j$ for every input node $j$, for the energy to remain finite.

Putting all the energy terms together, and denoting $E_{\rm input}(x,v)$ the energy term of input signals (either $E_{\rm input}^{\rm current}(x,v)$ or $E_{\rm input}^{\rm voltage}(x,v)$ depending on the case), we get the energy function of Eq.~(\ref{eq:energy-function-1}-\ref{eq:energy-function-2}).

\section{
Multi-Mechanism
Learning via Frequency Propagation}
\label{sec:frequency-propagation}

Learning in a resistive network consists in adjusting the branch conductances ($\theta$) so that the network exhibits a desired behaviour, i.e. a desired input-output function $x \mapsto v(\theta,x)$. In machine learning, this problem is formalized by introducing a \textit{cost function} $C$. Given an input-output pair $(x,y)$, the quantity $C(v(\theta,x),y)$ measures the discrepancy between the `model prediction' $v(\theta,x)$ and the desired output $y$. The learning objective is to find the parameters $\theta$ that minimize the expected cost $\mathbb{E}_{(x,y)} \left[ C(v(\theta,x),y) \right]$ over input-output pairs $(x,y)$ coming from the data distribution for the task that the system must solve.

%We can choose a cost function as follows. Let $y=(y_1,y_2,\ldots,y_K)$ be the desired output. We choose $2K$ of the nodes of the network, and split them in $K$ pairs $(v^T_{j+},v^T_{j-})$. We them define the cost function
%\begin{equation}
%    C(v,y) = \frac{1}{2} \sum_{k=1}^K (v^T_{k+}-v^T_{k-} - y_k)^2.
%\end{equation}
%as also proposed in \cite{kendall2020training,anisetti2022learning}.

In deep learning, the main tool for this optimization problem is stochastic gradient descent (SGD)~\cite{bottou2010large}: at each step we pick at random an example $(x,y)$ from the training set and update the parameters as
\begin{equation}
    \Delta \theta = - \eta \frac{\partial \mathcal{L}}{\partial \theta}(\theta,x,y),
\end{equation}
where $\eta$ is a step size, and
\begin{equation}
    \mathcal{L}(\theta,x,y) = C(v(\theta,x),y)
\end{equation}
is the per-example \textit{loss function}.

We now present \textit{frequency propagation} (Freq-prop), a learning algorithm for physical networks whose update rule performs SGD. Freq-prop proceeds by modifying the energy of the network to push or pull away the network's output values from the desired outputs. In the case of a resistive network (Section~\ref{sec:resistive-network}), we inject sinusoidal currents at the output nodes of the network, $i(t) = \gamma \sin(\omega t) \, \frac{\partial C}{\partial v}(v,y)$, where $t$ denotes time, $\omega$ is a frequency, and $\gamma$ is a small positive constant\footnote{In practical situations such as the squared error prediction, the cost function $C$ depends only on the state of output nodes ; therefore nudging requires injecting currents at output nodes only.}. This amounts to augment the energy function of the system by a time-dependent sinusoidal energy term $\gamma\, \sin(\omega t) \, C(v,y)$. Due to this perturbation, the system's response $v(t)$ minimizing the energy at time $t$ is
\begin{equation}
    \label{eq:response}
    v(t) = \underset{v}{\arg \min} \; \left[ E(\theta,x,v) + \gamma \, \sin(\omega t) \, C(v,y) \right].
\end{equation}
The response $v(t)$ is periodic of period $T=2\pi/\omega$, and for small perturbations (i.e. $\gamma \ll 1$), it is approximately sinusoidal. Next, we assume that we can recover the first two vectors of Fourier coefficients of $v(t)$, i.e. the vectors $a$ and $b$ such that
\begin{equation}
    \label{eq:definition-a-b}
    a = \frac{1}{T} \int_0^T v(t)dt, \qquad b = \frac{2}{T} \int_0^T v(t) \sin(\omega t) dt.
\end{equation}
Finally, denoting $a = (a_1,a_2,\ldots,a_N)$ and $b = (b_1,b_2,\ldots,b_N)$, we update each parameter $\theta_{jk}$ according to the learning rule
\begin{equation}
    \label{eq:learning-rule-frequence-prop}
    \Delta \theta_{jk} = - \alpha (b_j-b_k) \cdot (a_j-a_k),
\end{equation}
where $\alpha$ is a positive constant.

\begin{thm}
\label{thm:frequency-propagation}
For every parameter $\theta_{jk}$, we have
 \begin{equation}
    \Delta \theta_{jk} = - \alpha \, \gamma \frac{\partial \mathcal{L}}{\partial \theta_{jk}}(\theta,x,y) + O(\gamma^3)
\end{equation}
when $\gamma \to 0$.
\end{thm}

Namely, the learning rule \eqref{eq:learning-rule-frequence-prop} approximates one step of gradient descent with respect to the loss, with learning rate $\alpha \, \gamma$. Note that this learning rule is local: it requires solely locally available information for each parameter $\theta_{jk}$.

\begin{proof}[Proof of Theorem \ref{thm:frequency-propagation}]
Let $\theta$, $x$ and $y$ be fixed. For every $\beta \in \mathbb{R}$, we denote
\begin{equation}
    v_\star^\beta = \underset{v}{\arg \min} \; \left[ E(\theta,x,v) + \beta \, C(v,y) \right].
\end{equation}
With this notation, note that the response $v(t)$ of Eq.~\eqref{eq:response} rewrites $v(t) = v_\star^{\gamma\sin(\omega t)}$. Let us write the second-order Taylor expansion of $v_\star^\beta$ around $\beta=0$:
\begin{equation}
    \label{eq:taylor}
    v_\star^{\beta} = v_\star^0 + \beta \left. \frac{\partial v_\star^\beta}{\partial \beta} \right|_{\beta=0} + \frac{\beta^2}{2} \left. \frac{\partial^2 v_\star^\beta}{\partial \beta^2} \right|_{\beta=0} + O(\beta^3),
\end{equation}
where $v_\star^0 = v(\theta,x)$ by definition \eqref{eq:free-state}, and $\left. \frac{\partial v_\star^\beta}{\partial \beta} \right|_{\beta=0}$ and $\left. \frac{\partial^2 v_\star^\beta}{\partial \beta^2} \right|_{\beta=0}$ denote the derivative and second-derivative of $v_\star^\beta$ at $\beta=0$. Taking $\beta = \gamma \sin(\omega t)$ in the above formula, we get
\begin{align}
    v(t) = v_\star^{\gamma \sin(\omega t)} & = v_\star^0 + \gamma \sin(\omega t) \left. \frac{\partial v_\star^\beta}{\partial \beta} \right|_{\beta=0} \\
    & + \frac{\gamma^2}{2} \sin(\omega t)^2 \left. \frac{\partial^2 v_\star^\beta}{\partial \beta^2} \right|_{\beta=0} + O(\gamma^3),
\end{align}
uniformly in $t$. Therefore, the first two vectors of Fourier coefficients $a$ and $b$ of the periodic function $v(t)$, with time period $T=2\pi/\omega$ are
\begin{align}
    \label{eq:coefficient-a}
    a & =\frac{1}{T} \int_0^T v(t)dt = v_\star^0 + \frac{\gamma^2}{4} \left. \frac{\partial^2 v_\star^\beta}{\partial \beta^2} \right|_{\beta=0} + O(\gamma^3), \\
    b & =\frac{2}{T} \int_0^T v(t) \sin(\omega t) dt = \gamma \left. \frac{\partial v_\star^\beta}{\partial \beta} \right|_{\beta=0} + O(\gamma^3).
    \label{eq:coefficient-b}
\end{align}
Next, we know from the equilibrium propagation formula (Theorem 2.1 in \cite{scellier2021deep}) that the gradient of the loss $\mathcal{L}$ is equal to
\begin{equation}
    \frac{\partial \mathcal{L}}{\partial \theta}(\theta,x,y) = \left. \frac{d}{d\beta} \right|_{\beta=0} \frac{\partial E}{\partial \theta}(\theta,x,v_\star^\beta).
\end{equation}
Therefore,
\begin{equation}
    \frac{\partial \mathcal{L}}{\partial \theta}(\theta,x,y) = \frac{\partial^2 E}{\partial \theta \partial v}(\theta,x,v_\star^0) \cdot \left. \frac{\partial v_\star^\beta}{\partial \beta} \right|_{\beta=0}.
\end{equation}
Multiplying both sides by $\gamma$, and using \eqref{eq:coefficient-b}, we get
\begin{equation}
    \gamma \frac{\partial \mathcal{L}}{\partial \theta}(\theta,x,y) = \frac{\partial^2 E}{\partial \theta \partial v}(\theta,x,v_\star^0) \cdot b + O(\gamma^3).
\end{equation}
Finally, given the form of the energy function \eqref{eq:energy-function-1}, and using $b = O(\gamma)$ and $v_\star^0=a+O(\gamma^2)$ from Eq.~\eqref{eq:coefficient-a}, we get for every parameter $\theta_{jk}$
\begin{equation}
    \gamma \frac{\partial \mathcal{L}}{\partial \theta_{jk}}(\theta,x,y) = \left( a_j-a_k \right) \cdot \left( b_j-b_k \right) + O(\gamma^3).
\end{equation}
Therefore the learning rule
\begin{equation}
    \Delta \theta_{jk} = - \alpha (b_j-b_k) \cdot (a_j-a_k)
\end{equation}
satisfies
\begin{equation}
    \Delta \theta_{jk} = - \alpha \, \gamma \frac{\partial \mathcal{L}}{\partial \theta_{jk}}(\theta,x,y) + O(\gamma^3).
\end{equation}
Hence the result.
\end{proof}

\textbf{Remark 1.}
For simplicity, we have omitted the time of relaxation to equilibrium in our analysis. However, a practical circuit has an effective capacitance $C_{\rm eff}$ and therefore will equilibrate in time $\tau_{\rm relax} \sim R_{\rm eff} C_{\rm eff}$, where $R_{\rm eff}$ is the effective resistance of the circuit. Our learning algorithm will work as long as the circuit equilibrates much faster than the timescale of oscillation ($\tau_{\rm relax} \ll 1/\omega$). Our analysis thus requires that $C_{\rm eff}$ be small enough for the assumption $\tau_{\rm relax} \ll 1/\omega$ to hold. We leave the study of the regime where $C_{\rm eff}$ is non negligible for future work.

\textbf{Remark 2.}
While our nudging method \eqref{eq:response} is inspired by the one of \textit{equilibrium propagation} \cite{scellier2017equilibrium,kendall2020training}, it is also possible to apply the nudging variant of \textit{coupled learning} \cite{stern2021supervised} which might be easier to implement in practice \cite{dillavou2021demonstration}. To do this, we denote $v_O^F$ the `free' equilibrium value of the output nodes of the network (where the prediction is read), without nudging. Then, at time $t$, we \textit{clamp} the output nodes to $v_O^C(t) = v_O^F + \gamma \sin(\omega t) (y-v_O^F)$. This nudging method can be achieved via AC voltage sources at output nodes. We note however that Theorem~\ref{thm:frequency-propagation} does not hold with this alternative nudging method.

\textbf{Remark 3.}
Measuring $b_j$ for every node $j$ as per Eq.~\eqref{eq:definition-a-b} requires that we use the same reference time $t=0$ for all nodes, i.e. it requires global synchronization of the measurements for all nodes.
%Recovering the sign of $b_j$ requires that node $j$ knows the time at which the feedback signal $\gamma \, \sin(\omega t) \, C(v,y)$ was applied.
However, in practice, there may be a time delay $t_j$ between nudging and measurement, leading to a measured response $v_j(t) = a_j + b_j \sin(\omega (t+t_j)) + O(\gamma^3)$ at node $j$. Without any information about $t_j$, we can only obtain the absolute value of the coefficient $b_j$, not its sign. 
%If we use for each node $j$ only the shape of the response $v_j(\cdot)$, one can recover the magnitude of $b_j$ as $|b_j| = (\sup_t v_j(t) - \inf_t v_j(t)) / 2 + O(\gamma^2)$.
We propose a solution to this issue in section~\ref{sec:phase-information-problem}.

\section{Choice of the nudging signal}
\label{sec:phase-information-problem}

We have seen in section \ref{sec:frequency-propagation} that, using a sinusoidal nudging signal $\gamma \, \sin(\omega t) \, C(v,y)$, the measured response at node $j$ will be of the form $v_j(t) = a_j + b_j \sin(\omega (t+t_j)) + O(\gamma^3)$, where $t_j$ is the time delay between nudging and measurement. Unfortunately, it is not possible to recover the sign of $b_j$ without any knowledge of $t_j$. This problem can be overcome by using a different nudging signal.

In general, if we nudge the system by an energy term $\gamma f(t) \, C(v,y)$, where $f(t)$ is an arbitrary function such that $\sup_t |f(t)| < \infty$, then the system's response at node $j$ will be of the form $v_j(t) = a_j + b_j f(t+t_j) + O(\gamma^2)$.
%To see this, take $\beta=\gamma f(t)$ in Eq.~\ref{eq:taylor}.
Our goal is to choose a $f$ so that we can obtain for every node $j$ the values of $a_j$ and $b_j$ by measuring only $v_j(t)$, without knowing $t_j$.

Clearly, this is not possible for all functions $f$. For example, if $f(\cdot)$ is a constant, then $v_j(\cdot)$ is also a constant, and we cannot recover the values of $a_j$ and $b_j$ from $v_j(\cdot)$ alone. As seen above, another example for which this is not possible is $f(t)=\sin(\omega t)$. This is because a time delay $t_j=\pi/\omega$ will change the sign of the signal, $\sin(\omega (t+t_j))=-\sin(\omega t)$ ; therefore the sign of $b_j$ cannot be recovered without any knowledge of $t_j$.

An example of a nudging signal for which we can infer the values of $a_j$ and $b_j$ (up to $O(\gamma^2)$) is $f(t) = |\sin(\omega (t))|$. To do this, we observe the response at node $j$
\begin{equation}
    v_j(t) = a_j + b_j |\sin(\omega (t+t_j))| + O(\gamma^2)
\end{equation}
for a duration $\tau_{\mathrm{obs}}$ greater than the time period of the signal $T=2\pi/\omega$. The coefficients $a_j$ and $b_j$ can be obtained by identifying the times where the signal's derivative is zero or is discontinuous. Specifically, denoting $\partial_{+} v_j(t)$ and $\partial_{-} v_j(t)$ the left and right derivatives of the signal at time $t$, we have
\begin{align}
    & a_j = v_j(t_1) + O(\gamma^2) ~\mathrm{where}~\partial_{+} v(t_1) \neq \partial_{-} v(t_1), \\
    & b_j = v_j(t_2)-v_j(t_1) + O(\gamma^2) ~\mathrm{where}~ \partial v(t_2) = 0.
\end{align}

More generally, we show in Appendix \ref{sec:sign-problem} that, in principle, it is possible to recover the coefficients $a_j$ and $b_j$ if and only if the function $f$ has the property that there is no $\tau$ such that $f(t)=\sup{f}+\inf{f}-f(t+\tau)$ for every $t$. In other words, no amount of time delay converts the signal's `upright' form to its `inverted' form or vice versa.

%Though such a function would need a different procedure to extract the coefficients $a_j$ and $b_j$.

%If there is such a $\tau$, it is impossible to differentiate a signal with parameters $(a_j,b_j,t_j)$ from a signal with the parameters $(a_j+b_j(\sup{f}+\inf{f}),-b_j, t_j+\tau)$.

\section{
General Applicability of
Frequency Propagation
%in Other Physical Networks
}

%There have been growing interest in training mechanical networks for desired functionality\citep{stern2021supervised,murgan_stern_folding,stern_continual}.
Freq-prop applies to arbitrary physical networks: not only resistive networks,
%(Section~\ref{sec:resistive-network}),
but also flow networks,
%(Appendix~\ref{sec:chemical-signaling}),
capacitive networks and inductive networks, %~\cite{cherry1951cxvii},
among others. In these networks, the notion of current-voltage characteristics will be replaced by current-pressure characteristics, current-flux characteristics, and charge-voltage characteristics, respectively. The mathematical framework for nonlinear elements (Section~\ref{sec:resistive-network}) also applies to these networks, where the energy functions minimized at equilibrium are the co-content, the inductive energy and the capacitive co-energy, respectively~\cite{millar1951cxvi,cherry1951cxvii}.

To emphasize the generality of Freq-prop, we present it here in the context of \textit{central force spring networks} (or `elastic networks') \citep{stern2021supervised}. We consider an elastic network of $N$ nodes interconnected by springs. The elastic energy stored in the spring connecting node $i$ to node $j$ is $E_{ij}(r_{ij}) = \frac{1}{2} k_{ij} \left( r_{ij} - \ell_{ij} \right)^2$, where $k_{ij}$ is the spring constant, $\ell_j$ is the spring's length at rest, and $r_{ij}$ is the distance between nodes $i$ and $j$. Nonlinear springs are also allowed and their energy terms are gathered in a unique term $E_{\rm nonlinear}(r)$. Thus, the total elastic energy stored in the network, which is minimized, is given by
\begin{equation}
    E(\theta,r) = \frac{1}{2} \sum_{i,j} k_{ij} \left( r_{ij} - \ell_{ij} \right)^2 + E_{\rm nonlinear}(r),
\end{equation}
where $\theta = \{ k_{ij}, \ell_{ij} \}$ is the set of adjustable parameters, and $r = \{ r_{ij} \}$ plays the role of state variable.

In this setting as in the case of resistive networks, we apply a nudging signal $\gamma\, \sin(\omega t) \, C(r,y)$ at the output part of the network, we observe the response $r(t)$, and we assume that we can recover the first two vectors of Fourier coefficients of $r(t)$, i.e. the vectors $a$ and $b$ such that $a = \frac{1}{T} \int_0^T r(t)dt$ and $b = \frac{2}{T} \int_0^T r(t) \sin(\omega t) dt$. Then, the learning rules for the spring constant $k_{ij}$ and the spring's length at rest $\ell_{ij}$ read, in this context,
\begin{equation}
    \Delta k_{ij} = - \alpha \, b_{ij} \, (a_{ij}-\ell_{ij}), \qquad \Delta \ell_{ij} = - \alpha \, k_{ij} \, b_{ij}.
\end{equation}
Theorem \ref{thm:frequency-propagation} generalizes to this setting ; the above learning rules perform stochastic gradient descent on the loss: $\Delta \theta = - \alpha \gamma \frac{\partial \mathcal{L}}{\partial \theta}(\theta,x,y) + O(\gamma^3)$.

\section{Related Work}

Frequency propagation builds on learning via \textit{chemical signaling} \cite{anisetti2022learning}, which is another example of multi-mechanism learning (MmL) in physical networks. Whereas MmL via frequency propagation uses two different frequencies to play the role of the \textit{activation} and \textit{error} signals during training, MmL via chemical signaling uses two different chemical concentrations for these signals. \cite{anisetti2022learning} presents learning via chemical signaling in the setting of linear flow networks, which we extend here to the nonlinear setting (Appendix~\ref{sec:chemical-signaling}).

Freq-prop is also related to \textit{equilibrium propagation} (EP) \cite{scellier2017equilibrium,kendall2020training,scellier2021deep} and \textit{coupled learning} \cite{stern2021supervised}. To see the relationship with these algorithms, we consider the case of resistive networks (section \ref{sec:resistive-network}). Denote $v_{jk} = v_j-v_k$ the voltage across branch $(j,k)$. Further denote $v^\beta = \arg\min_{v} \; \left[ E(\theta,x,v) + \beta \, C(v,y) \right]$ for any $\beta \in \mathbb{R}$. Based on a result from \cite{scellier2017equilibrium}, \cite{kendall2020training} proved that the learning rule
\begin{equation}
    \label{eq:eqprop-learning-rule}
    \Delta^{\rm EP} \theta_{jk} = \frac{\alpha}{2} \left( (v^0_{jk})^2 - (v^\beta_{jk})^2  \right)
\end{equation}
performs gradient descent with step size $\alpha \beta$, up to $O(\beta^2)$. We note that the right-hand side of \eqref{eq:eqprop-learning-rule} is also equal to $\alpha \, v^0_{jk} \left( v^0_{jk} - v^\beta_{jk} \right) + O(\beta^2)$, showing that the gradient information is contained in the physical quantities $v^0$ and $\left. \frac{\partial v^\beta}{\partial \beta} \right|_{\beta=0}$. These quantities correspond to the activation and error signals of Freq-prop, respectively. To avoid the use of finite differences to measure $\left. \frac{\partial v^\beta}{\partial \beta} \right|_{\beta=0}$, Freq-prop makes use of a time-varying nudging signal $\beta(t) = \gamma \sin(\omega t)$. With this method, the activation and error signals are encoded in the frequencies $0$ and $\omega$ of the response signal $v(t) = v^0 + \gamma \sin(\omega t) \left. \frac{\partial v^\beta}{\partial \beta} \right|_{\beta=0} + O(\gamma^3)$. The required information can thus be recovered via frequency analysis.
%since the two vectors of Fourier coefficients of $v(t)$ are $a=v^0 + O(\gamma^2)$ and $b=\gamma \left. \frac{\partial v^\beta_j}{\partial \beta} \right|_{\beta=0} + O(\gamma^2)$.

Very recent work proposes an alternative approach to train physical systems by gradient descent called \textit{agnostic equilibrium propagation}~\cite{scellier2022agnostic}. However, this method imposes constraints on the nature of the parameters ($\theta$), which must minimize the system's energy ($E$), just like the state variables ($v$) do. This assumption does not allow us to view the conductances of resistors as trainable parameters in a resistive network. The method also requires control knobs with the ability to perform homeostatic control over the parameters.
Our work can also be seen as a physical implementation of \textit{implicit differentiation} in physical networks. We refer to \cite{zucchet2022beyond} for a description of implicit differentiation where the authors use a mathematical formalism close to ours.

Lastly, other physical learning algorithms that make explicit use of time are being developed. For instance, recent work proposes a way to train physical systems with time reversible Hamiltonians \cite{lopez2021self}. In this method called \textit{Hamiltonian echo backpropagation} (HEB), the error signal is a time-reversed version -- an ``echo" -- of the activation signal, with the cost function acting as a perturbation on this signal.
%This signal carries the output error information back into the system. However, in MmL the cost function by itself is seen as a source of the error signal, and it doesn't have to be encoded as a perturbation on the activation signal. HEB could potentially have non-linear optic applications in neuromorphic computing.
However, HEB requires a feasible way to time-reverse the activation signal.

\section{Discussion}

We have introduced frequency propagation (Freq-prop), a physical learning algorithm that falls in the category of Multi-mechanism Learning (MmL). In MmL, separate and ``distinguishable" activation and error signals both contribute to a local learning rule, such that trainable parameters (e.g. conductances of variable resistors) perform gradient descent on a loss function.
%For MmL to work, the error signal should not ``mix" with the activation signal, or, in other words, be ``distinguishable" from the activation signal.
In Freq-prop, the activation and error signals are implemented using different frequencies of a single physical quantity (e.g. voltages or currents) and are thus distinguishable.
%In Freq-prop, distinguishability can, in principle, be achieved, for example, with direct current activation signal and an alternating current for the error signal in limit of small nudging in a non-linear resistor network.
We note however that the `distinguishability' of the signals does not mean that they are mathematically `independent':
%As for the relationship between the error and the activation signal,
in Freq-prop, the error signal depends on the activation signal via the Hessian of the network.
%, which, for a non-linear network may depend on the activation signal, due to which, changing the activation signal may change the error signal.
Other potential MmL algorithms may involve independent physical mechanisms, such as an electrical activation signal and a chemical error signal or vice versa.

Multi-mechanism learning algorithms, such as Freq-prop, may have implications towards designing fast and low-power, or high-efficiency, hardware for AI, as they are rooted in physical principles. For the time being, inroads are being made by using backpropagation to train controllable physical systems in a hybrid {\it in silico-in situ}  approach~\cite{wright2022deep}. As we work towards a fully in situ approach, Freq-prop is a natural candidate. And while the {\it in situ} realization of a nonlinear resistor network is an obvious starting point, there are potential limitations, particularly in terms of timescales. Consider the time of relaxation to equilibrium ($\tau_{\rm relax}$), the time scale of the sinusoidal nudging signal ($T = 2\pi/\omega$), and the time scale of learning ($\tau_{\rm learning}$). Our learning methodology requires that $\tau_{\rm relax} \ll T < \tau_{\rm learning}$. More specifically,
\begin{enumerate}
    \item Once input is applied, the network reaches equilibrium in time $\tau_{\rm relax}$.
   \item Based on the network's output, a sinusoidal nudging signal of frequency $\omega$ is applied at the output nodes. The time scale of evolution of this sinusoidal nudging wave is $T = 2\pi/\omega$. Assuming that $\tau_{\rm relax} \ll T$, the network is at equilibrium at every instant $t$. 
    \item We observe the network's response $v(t)$ for a time $\tau_{\rm obs} > T$ to extract the coefficients $a$ and $b$ of Eq.~\eqref{eq:definition-a-b}. Updating the conductances of the resistors takes a time $\tau_{\rm learning} \sim \tau_{\rm obs}$ using the values of $a$ and $b$ to determine the magnitude and sign of these updates.
\end{enumerate}

Finally, could something like Freq-prop occur in the brain? Earlier work analyzing local field potentials recorded simultaneously from different regions in the cortex suggested that feedforward signaling is carried by gamma-band (30–80 Hz) activity, whereas feedback signaling is mediated by alpha-(5–15Hz) or beta- (14–18 Hz) band activity, though local field potentials are not actively relayed between regions~\cite{bastos2015}. More recent work in the visual cortex argues that feedforward and feedback signaling rely on separate “channels” since correlations in neuronal population activity patterns, which are actively relayed between regions, are distinct during feedforward- and feedback-dominated periods~\cite{semedo2022}. Freq-prop is also related in spirit to the idea of frequency multiplexing in biological neural networks~\cite{naud2018sparse,payeur2021burst,akam_kullmann_2014}, which uses the simultaneous encoding of two or more signals.
%In Freq-prop, the separate activation and error signals are encoded in a single channel.
While Freq-prop here uses only two separate signals -- an activation signal and an error signal -- one can envision multiple activation and error signals being encoded to accommodate vector inputs and outputs and to accommodate multiple, competing cost functions. With multiple activation and error signals one can also envision coupling learning via chemical signaling (Appendix~\ref{sec:chemical-signaling}) with Freq-prop, for example, to begin to capture the full computational {\it creativity} of the brain.

%Comparison to Learning as a least-control problem and other approaches to learning. 

%TODO: we could in principle mix in principle 3 different kinds of signals (using three different frequencies). But is there any obvious way how this could be useful?

\acknowledgements{The authors thank Sam Dillavou, Nachi Stern and Jack Kendall for discussion. JMS acknowledges financial support from NSF-DMR-1832002.}

\bibliography{biblio.bib} 
\bibliographystyle{ieeetr}

\appendix

\section{Further Details on the Nudging Signal}
\label{sec:sign-problem}

Let $f(t)$ denote the nudging signal. Assuming that $f$ is bounded, recall that, for every $j$, the measured response $v_j(t)$ at node $j$ is of the form $v_j(t)=a_j+b_j f(t+t_j) + O(\gamma^2)$, where $a_j$ and $b_j$ are the numbers that we wish to recover (up to $O(\gamma^2)$) to implement the parameter update, and $t_j$ is an unknown time delay. Our goal is to obtain for every node $j$ the values of $a_j$ and $b_j$ by measuring only $v_j(t)$, without any knowledge of $t_j$.

We now establish a necessary and sufficient condition on the nudging signal $f(t)$ so that one can, at least in principle, uniquely obtain the values of $a_j$ and $b_j$ for every node $j$. We are concerned with quantities that depend only on a single node and hence we will drop the node index with the understanding that all of the analysis applies to any arbitrary node.

Let $F$ denote the set of all real-valued, bounded functions,
and let $f$ be an element of $F$.
Let $\mathcal{C}_f:\mathbb{R}^3 \rightarrow F$ be the function that maps the parameters $(a,b,t_0)$ to the function $v(\cdot) = a + bf(\cdot+t_0)$.
We define the following equivalence relation on $F$: two functions $g, h \in F$ are equivalent if they differ by a time translation, i.e., $g \sim h$ if and only if there exists a $t_0 \in \mathbb{R}$ such that $g(t)=h(t+t_0)$ for all $t \in \mathbb{R}$. Let $\tilde{F}=F/\sim$ be the quotient of $F$ under this equivalence relation and let $[g]$ be the equivalence class that contains the function $g$.
The map $\mathcal{C}_f$ can be lifted to yield $\mathcal{\tilde{C}}_{f}:\mathbb{R}^2 \rightarrow \tilde{F}$ such that $\tilde{\mathcal{C}}_{f}(a,b)= [a+bf]$. In order to be able to uniquely extract $a$ and $b$ from any equivalence class of the form $[a+bf]$, the function $\tilde{\mathcal{C}}_{f}$ has to be injective. This can be re-expressed as a direct condition on the nudging signal $f$.

\begin{prop}
The following statements are equivalent:
\begin{compactenum}
\item [P1:] The function $\mathcal{\tilde{C}}_{f}:\mathbb{R}^2 \rightarrow \tilde{F}$ defined by $\tilde{\mathcal{C}}_{f}(a,b)=[a+bf]$ is injective.
\item [P2:]  There exists no $\tau \in \mathbb{R}$ such that for all $t\in \mathbb{R}$,\\ $f(t) = \sup f + \inf f - f(t+\tau)$.
\end{compactenum}
where $\sup{f} = \sup_t f(t)$ and $\inf{f}=\inf_t f(t)$ denote the supremum and infimum values of the nudging signal $f$ respectively.
\end{prop}

\begin{proof}
We establish this by proving that the negation of the two statements are equivalent, i.e., the following statements are equivalent:
\begin{compactenum}
    \item [N1:]
    %The function $\mathcal{\tilde{C}}_{f}:\mathbb{R}^2 \rightarrow \tilde{F}$ defined by $\tilde{\mathcal{C}}_{f}(a,b)=[a+bf]$ is not injective.
    There exists two distinct pairs of real numbers $(a_1,b_1)$ and $(a_2,b_2)$ such that $[a_1+b_1f]=[a_2+b_2f]$.
    \item [N2:] There exists a $\tau \in \mathbb{R}$ such that for all $t\in \mathbb{R}$,\\ $f(t)=\sup{f}+\inf{f}-f(t+\tau)$.
\end{compactenum}

Suppose that N2 is true: there is a $\tau \in \mathbb{R}$ such that for all $t\in \mathbb{R}$, $f(t)=\sup{f}+\inf{f}-f(t+\tau)$. This means that $f$ and $\sup{f}+\inf{f}-f$ are related by a time translation, i.e. $[f]=[\sup{f}+\inf{f}-f]$. Therefore, N1 is true, with $(a_1,b_1) = (0,1)$ and $(a_2,b_2) = (\sup{f}+\inf{f},-1)$.

Conversely, suppose that N1 is true:
%there must be two or more distinct parameter pairs that give rise to signals that are related to each other by time translations, i.e.,  More explicitly,
there exists two distinct pairs of real numbers $(a_1,b_1)$ and $(a_2,b_2)$ and a $\tau\in\mathbb{R}$ such that
\begin{equation}
    \forall t\in \mathbb{R}, \qquad a_1+b_1f(t)=a_2+b_2f(t+\tau).
\end{equation}
The numbers $b_1$ and $b_2$ cannot be both zero, otherwise the above equation implies that $a_1=a_2$, a contradiction. If $b_1=0$ and $b_2 \neq 0$, the above equation implies that $f$ is a constant, in which case N2 is clearly true. Otherwise $b_1 \neq 0$ and we can re-write the above equality as
\begin{equation}
    \label{eq:details-0}
    \forall t\in \mathbb{R}, \qquad f(t)=a+bf(t+\tau)
\end{equation}
with $a=(a_2-a_1)/b_1$ and $b=b_2/b_1$. Now there are two possibilities: either $b>0$ or $b<0$.

First, let us suppose that $b>0$. The above equality imposes the following conditions on the minimum and maximum values of the function $f$:
\begin{align}
    \label{eq:details-1}
    \sup f & = a + b \sup f,\\
    \label{eq:details-2}
    \inf f & = a + b \inf f.
\end{align}
Subtracting \eqref{eq:details-2} from \eqref{eq:details-1} and reorganizing the terms we get $(1-b) (\sup f - \inf f) = 0$. If $b=1$, then $a=0$, contradicting our assumption that $(a_1,b_1)$ and $(a_2,b_2)$ are distinct pairs. Therefore $\sup f = \inf f$, $f$ is constant and N2 is true.

Second, let us suppose that $b<0$. As before we have
\begin{align}
    \sup f & = a + b \inf f,\\
    \inf f & = a + b \sup f,
\end{align}
and again $(1+b) (\sup f - \inf f) = 0$. Either $f$ is a constant, or $b=-1$, impliying in turn that $a = \sup f + \inf f$. Therefore, coming back to \eqref{eq:details-0}, we have $f(t)=\sup{f}+\inf{f}-f(t+\tau)$ for all $t\in \mathbb{R}$, which is the statement of N2.
\end{proof}

\section{Multi-Mechanism Learning via Chemical Signaling}
\label{sec:chemical-signaling}

In this appendix, we generalize the learning algorithm via \textit{chemical signaling} \cite{anisetti2022learning} to nonlinear networks. 
Learning via chemical signaling is another example of \textit{multi-mechanism learning} in physical networks. It uses pressures and chemical concentrations to implement a local learning rule. This way of using multiple independent ``mechanisms" is the central idea behind multi-mechanism learning.

Consider a flow network, i.e. a network of nodes interconnected by tubes. A flow network is formally equivalent to the resistive network of Section \ref{sec:resistive-network}, with $v$ being the configuration of node pressures, and $\theta_{jk}$ being the conductance of the branch between nodes $j$ and $k$.

Learning via chemical signaling proceeds as follows. In the first phase, given $\theta$ and input signals $x$, the configuration of node pressures stabilizes to its equilibrium value $v(\theta,x)$ given by
%Eq.~\eqref{eq:free-state}.
\begin{equation}
    v(\theta,x) = \underset{v}{\arg \min} \; E(\theta,x,v).
\end{equation}
In the second phase, we inject chemical currents $e = - \beta \frac{\partial C}{\partial v}(v(\theta,x),y)$ at output nodes, where $\beta$ is a (positive or negative) constant. As a result, a chemical concentration $u$ develops at each node. Assuming that the configuration of node pressures $v(\theta,x)$ is not affected by the chemical, the chemical concentration $u$ at equilibrium satisfies the relationship:
\begin{equation}
    \label{eq:chemical-concentration}
    \frac{\partial^2 E}{\partial v^2}(\theta,x,v(\theta,x)) \cdot u = - \beta \frac{\partial C}{\partial v}(v(\theta,x),y).
\end{equation}
%\begin{equation}]
%    u_\star^\beta = \underset{u}{\arg \min} \left[ E(\theta,x,v(\theta,x)+u) + \beta C(v(\theta,x)+u,y) \right].
%\end{equation}
Indeed, diffusion along a tube follows the same equation as that of flow along the same tube, up to a constant factor (replacing node pressures and flow conductivity by chemical concentration and diffusion constant, respectively).
When there is no ambiguity from the context, we write $v = v(\theta,x)$ for simplicity. We note that, although $v$ is not affected by the chemical, $u$ depends on $v$. In particular $u$ also depends on $\theta$ and $x$ through $v$.

%\textbf{TODO: a) justify a little that u should satisfy this equation, and b) add a multiplicative factor in front of u, so that v and u have the same physical dimensions (otherwise it doesn't make sense to add v and u, physically speaking).}

Next, denoting $u = (u_1,u_2,\ldots,u_N)$, we update each parameter $\theta_{jk}$ according to the learning rule
\begin{equation}
    \label{eq:learning-rule-chemical-signaling}
    \Delta \theta_{jk} = - \alpha (u_j-u_k) \cdot (v_j-v_k),
\end{equation}
where $\alpha$ is some constant. Note that this learning rule is local (just like the learning rule of Freq-prop), requiring only information about nodes $j$ and $k$ for each conductance $\theta_{jk}$.

\begin{thm}
\label{thm:chemical-signaling}
For every parameter $\theta_{jk}$, it holds that
\begin{equation}
    \Delta \theta_{jk} = - \alpha \, \beta \frac{\partial \mathcal{L}}{\partial \theta_{jk}}(\theta,x,y).
\end{equation}
\end{thm} 

Namely, the learning rule of Eq.~\eqref{eq:learning-rule-chemical-signaling} performs one step of gradient descent with respect to the loss, with step size $\alpha \beta$.
We note that learning via chemical signaling comes in two variants, either with $\beta > 0$ and $\alpha > 0$, or with $\beta < 0$ and $\alpha < 0$. The procedure performs one step of gradient \textit{descent} as long as the product $\alpha \beta$ is positive.

\bigskip

\begin{proof}[Proof of Theorem \ref{thm:chemical-signaling}]
First, we write the first-order equilibrium condition for $v(\theta,x)$, which is
\begin{equation}
    \label{eq:equilibrium-v}
    \frac{\partial E}{\partial v}(\theta,x,v(\theta,x)) = 0.
\end{equation}
We differentiate this equation with respect to $\theta$:
\begin{equation}
    \frac{\partial^2 E}{\partial v^2}(\theta,x,v(\theta,x)) \frac{\partial v}{\partial \theta}(\theta,x) + \frac{\partial^2 E}{\partial v \partial \theta}(\theta,x,v(\theta,x)) = 0.
\end{equation}
Multiplying both sides on the left by $u^\top$ we get
\begin{equation}
    \label{eq:proof-1}
    u^\top \frac{\partial^2 E}{\partial v^2}(\theta,x,v(\theta,x)) \frac{\partial v}{\partial \theta}(\theta,x) + u^\top \frac{\partial^2 E}{\partial v \partial \theta}(\theta,x,v(\theta,x)) = 0.
\end{equation}
On the other hand, multiplying both sides of \eqref{eq:chemical-concentration} on the left by $\frac{\partial v}{\partial \theta}(\theta,x)^\top$, we get
\begin{align}
    \frac{\partial v}{\partial \theta}(\theta,x)^\top \frac{\partial^2 E}{\partial v^2}(\theta,x,v(\theta,x)) u & = - \beta \frac{\partial v}{\partial \theta}(\theta,x)^\top \frac{\partial C}{\partial v}(v(\theta,x),y) \\
    & = - \beta \frac{\partial \mathcal{L}}{\partial \theta}(\theta,x,y)
    \label{eq:proof-2}
\end{align}
Comparing \eqref{eq:proof-1} and \eqref{eq:proof-2} we conclude that
\begin{equation}
    u^\top \frac{\partial^2 E}{\partial v \partial \theta}(\theta,x,v(\theta,x)) = \beta \frac{\partial \mathcal{L}}{\partial \theta}(\theta,x,y).
\end{equation}
Finally, using the form of the energy function \eqref{eq:energy-function-1}, we have for each parameter $\theta_{ij}$
\begin{equation}
    (u_i-u_j) \cdot (v_i-v_j) = \beta \frac{\partial \mathcal{L}}{\partial \theta_{ij}}(\theta,x,y).
\end{equation}
Therefore the learning rule
\begin{equation}
    \Delta \theta_{jk} = - \alpha (u_i-u_j) \cdot (v_i-v_j)
\end{equation}
satisfies
\begin{equation}
    \Delta \theta_{jk} = - \alpha \, \beta \frac{\partial \mathcal{L}}{\partial \theta_{jk}}(\theta,x,y). % + O(\gamma^2).
\end{equation}
Hence the result.
\end{proof}

\end{document}